\documentclass[a4paper,12pt]{amsproc}
\usepackage{amsthm}
\usepackage{amsmath,amssymb,amsthm}
\usepackage{amsaddr}
\usepackage{amsfonts}
\usepackage{braket}
\usepackage{chngcntr}
\usepackage{subcaption}
\counterwithin{figure}{section}
\usepackage[T1]{fontenc}
\usepackage{hyperref}
\usepackage{tikz}
\usepackage{graphicx}
\usepackage{mathtools}
\newtheorem{thm}{Theorem}
\newtheorem{lem}{Lemma}
\newtheorem{defn}{Definition}
\newtheorem{prop}{Proposition}
\newtheorem{rem}{Remark}
\theoremstyle{definition}
\newtheorem{defn2}{Example}
\newenvironment{defnbis}[1]
  {\renewcommand{\thedefn}{\ref{#1}$'$}%
   \addtocounter{defn}{-1}%
   \begin{defn}}
  {\end{defn}}
\newcommand{\Complex}{\mathbb{C}}
\newcommand{\Rational}{\mathbb{Q}}

\newcommand{\tree}{%
  \vcenter{\hbox{\tikz[node distance=2.5ex]{%
     \draw[thick] (5,-0.20) -- (5,0) -- (4.85,0.10) -- (5,0) -- (5.15,0.10);
}}}}
\newcommand{\bridge}{\raisebox{1.75mm}{$\frown$}\hspace{-4.5mm}\vee}
\newcommand{\treetwo}{%
	\vcenter{\hbox{\tikz[node distance=2.5ex]{%
				\draw(1,0.75) -- (1,1.25);\draw (1,0.75) node{ $\bullet$};\draw (1,1.25) node{ $\bullet$};
}}}}
\newcommand{\treethree}{%
	\vcenter{\hbox{\tikz[node distance=2.5ex]{%
					\draw[thick](1,-1)--(1,0) -- (1,1);\draw (1,-1) node{ $\bullet$};\draw (1,0) node{ $\bullet$};\draw (1,1) node{ $\bullet$};
}}}}
\newcommand{\treethreetwo}{%
	\vcenter{\hbox{\tikz[node distance=2.5ex]{%
				\draw[thick](-0.7,0)--(0,1) -- (0.7,0);\draw (-0.7,0) node{ $\bullet$};\draw (0,1) node{ $\bullet$};\draw (0.7,0) node{ $\bullet$};
}}}}
\newcommand{\oneloop}{%
	\vcenter{\hbox{\tikz[node distance=2.5ex]{%
				\draw[thick] (5,-0.20) -- (5,0) -- (4.85,0.10) -- (5,0) -- (5.15,0.10) ;\draw[thick] (5.15,0.13) arc (45:135:0.2cm);
}}}}
\newcommand{\onetwo}{%
	\vcenter{\hbox{\tikz[node distance=2.5ex]{%
				\draw[thick] (6.25,0.75) -- (6.25,1) -- (6,1.25) -- (6.125,1.125)--(6.25,1.25)--(6.125,1.125) -- (6.25,1) -- (6.5,1.25);
}}}}
\newcommand{\twoone}{%
	\vcenter{\hbox{\tikz[node distance=2.5ex]{%
				\draw[thick] (6.25,0.75) -- (6.25,1) -- (6,1.25) -- (6.125,1.125) -- (6.25,1) -- (6.375,1.125)--(6.25,1.25)--(6.375,1.125)--(6.5,1.25);
}}}}
\newcommand{\onetwoloopone}{%
\vcenter{\hbox{\tikz[node distance=2.5ex]{%
\draw[thick] (6.25,0.75) -- (6.25,1) -- (6,1.25) -- (6.125,1.125)--(6.25,1.25)--(6.125,1.125) -- (6.25,1) -- (6.5,1.25);\draw[thick] (6.25,1.3) arc (45:135:0.18cm);
}}}}
\newcommand{\onetwolooptwo}{%
	\vcenter{\hbox{\tikz[node distance=2.5ex]{%
				\draw[thick] (6.25,0.75) -- (6.25,1) -- (6,1.25) -- (6.125,1.125)--(6.25,1.25)--(6.125,1.125) -- (6.25,1) -- (6.5,1.25);\draw[thick] (6.5,1.3) arc (45:135:0.15cm);
}}}}
\newcommand{\twooneloopone}{%
	\vcenter{\hbox{\tikz[node distance=2.5ex]{%
				\draw[thick] (6.25,0.75) -- (6.25,1) -- (6,1.25) -- (6.125,1.125) -- (6.25,1) -- (6.375,1.125)--(6.25,1.25)--(6.375,1.125)--(6.5,1.25);\draw[thick] (6.23,1.3) arc (45:135:0.17cm);
}}}}
\newcommand{\twoonelooptwo}{%
	\vcenter{\hbox{\tikz[node distance=2.5ex]{%
				\draw[thick] (6.25,0.75) -- (6.25,1) -- (6,1.25) -- (6.125,1.125) -- (6.25,1) -- (6.375,1.125)--(6.25,1.25)--(6.375,1.125)--(6.5,1.25);\draw[thick] (6.5,1.3) arc (45:135:0.15cm);
}}}}
\begin{document}
\title{Hopf Algebras and Topological Recursion}
\author{Jo\~{a}o N. Esteves}
\address{CAMGSD, Departamento de Matem\'{a}tica, Instituto Superior T\'{e}cnico, Av. Rovisco Pais 1, 1049-001 Lisboa, Portugal}
\email{joao.n.esteves@tecnico.ulisboa.pt}
\thanks{The author was supported by Funda\c{c}\~ao para a Ci\^{e}ncia e a Tecnologia through the grant SFRH/BPD/77123/2011. He greatly thanks the organizers of the 2016 von Neumann Symposium for the opportunity of giving this talk.}
\begin{abstract}
We first review our previous work where we considered a model for topological recursion based on the Hopf Algebra of planar binary trees of Loday and Ronco and showed that extending this Hopf Algebra by identifying pairs of nearest neighbor leaves and thus producing graphs with loops we obtain the full recursion formula of Eynard and Orantin. Then we discuss the algebraic structure of the spaces of correlation functions in $g=0$ and in $g>0$. By taking a classical and a quantum product respectively we endow both spaces with a ring structure. 

This is an extended version of the contributed talk given at the 2016 von Neumann Symposium on Topological Recursion and its Influence in Analysis, Geometry and Topology, from 4 to 8 July 2016 at Hilton Charlotte University Place, USA.
\end{abstract}

\maketitle
\tableofcontents

\section{Introduction}
 In the paper \cite{1751-8121-48-44-445205} we showed how the Hopf Algebra of planar binary trees of Loday and Ronco \cite{MR1654173} can be seen as a sort of representation of the space generated by correlation functions that obey the Eynard-Orantin recursion formula. These correlation functions are graded by the Euler characteristic and we can consider for each degree the vector space over a field $k$ of characteristic 0 generated by them and then take the direct sum of these vector spaces for all degrees. Here we will discus, among other things and after reviewing the work in \cite{1751-8121-48-44-445205} if this vector space can be given additional algebraic structures. 
 
 First we consider planar binary trees of order $n$, that is with $n$ vertices's and $n+1$ leaves, as a representation of genus $g=0$ correlation function $W_{k}^0(p,p_1,\dots, p_{k-1})$ of Euler characteristic $\chi=2-2g-k$ equal to $-n$. Here the Euler characteristic is the one of Riemann or topological surfaces of genus $g$ and $k$ punctures or borders to which the correlation functions $W_k^g(p,p_1,\dots, p_{k-1})$ are usually related in some concrete problems. For $g=0$ we label the root with $p$ and the $n+1$ leaves with the $p_1,\dots,p_{n+1}$ variables. Then by connecting the nearest neighbors leaves with a single edge and reducing the number of pairs of labels in the same way as increasing the genus we obtain graphs with loops that we see as a representation of higher genus correlation functions with the same Euler characteristic. 
 
 As an example take $W^2_1(p)$ which has $\chi=-3$. Its underline generating trees are planar binary trees of order 3 which are also models for $W_5^0(p,p_1,p_2,p_3,p_4)$:
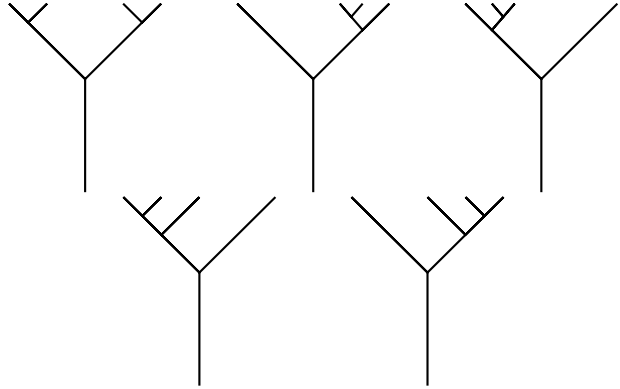
\begin{figure}
  \begin{tikzpicture}\label{fig:pbt5a}
  \draw[thick](1,-0.5) -- (1,1) -- (0.25,1.75) -- (0.5,2)  -- (0.25,1.75) -- (0,2) -- (1,1) -- (2,2) -- (1.75,1.75) -- (1.5,2);\draw[thick] (4,-0.5) -- (4,1) -- (3,2) -- (4,1) -- (5,2) -- (4.65,1.65) -- (4.35,2)-- (4.5,1.825)--(4.65,2);\draw[thick] (7,-0.5) -- (7,1) -- (6,2) -- (6.35,1.65)--(6.65,2)--(6.5,1.825)-- (6.35,2)--(6.5,1.825)--(6.35,1.65)-- (7,1) -- (8,2);
  \end{tikzpicture}
  \begin{tikzpicture}\label{fig:pbt5b}
  \draw[thick] (10,-0.5) -- (10,1) -- (9,2) -- (9.25,1.75)--(9.5,2)--(9.25,1.75) --  (9.5,1.5) -- (10,2) -- (9.5,1.5) -- (10,1) -- (11,2);
  \draw[thick] (13,-0.5) -- (13,1) -- (12,2) -- (13,1) -- (14,2)--(13.75,1.75)--(13.5,2)--(13.75,1.75) --  (13.5,1.5) -- (13,2) -- (13.5,1.5) ;
  \end{tikzpicture}
  \caption{Planar binary trees of order 3 as generators of the correlation functions $W^0_5, W^1_3$ and $W^2_1$ with $\chi=-3$. Reprinted with permission from \cite{1751-8121-48-44-445205}.}\label{fig:pbt5}
\end{figure}
 Identifying pairs of nearest neighbor leaves in the left and right branches independently we get the second term of topological recursion. Identifying pairs of leaves each taken from the left and the right branches gives the first term. Note that in this case not every planar binary tree of order 3 gives $W_1^2$. In fact the first tree of fig. \ref{fig:pbt5} does not give a genus 2 correlation function by identifying the nearest neighbor leaves in opposite branches.
 \section{The topological recursion of Eynard and Orantin}
 The topological recursion formula of Eynard and Orantin has its origin in Matrix Models, for general reviews see for instance \cite{DiFrancesco:1993nw,MR2346575}. In the hermitian 1-matrix form of the theory the purpose is to compute connected correlation functions $W_{k+1}$ depending on a set of variables $p, p_1,\dots,p_k$
 \begin{equation}
 W_{k+1}(p,p_1,\dots p_k)=\Braket{\text{Tr}\frac{1}{p-M}\text{Tr}\frac{1}{p_1-M}\dots \text{Tr}\frac{1}{p_k-M}}_c
 \end{equation}
 starting with $W_1(p)$ and $W_2(p,p_1)$.
  These functions which are solutions of the so-called loop equations are only well defined over Riemann surfaces because in $\Complex$ they are multi-valued. They admit an expansion on the order $N$ of the random matrix $M$, with components $W_{k+1}^g(p,p_1,\dots p_k)$ related to a definite genus. We will not be concerned here with the actual computation of correlation functions in specific models.
  
  Let $K=(p_1,\dots,p_k)$ be a vector of variables. For instance in concrete cases these can be coordinates of punctures on Riemann surfaces, labels of borders on topological surfaces or variables in Matrix Models, but we just leave them as labels of leaves of planar binary trees or of graphs obtained from planar binary trees. We assign the label $p$ to the root of a tree or of a graph with loops obtained from a tree. The topological recursion formula is
 \begin{align}\label{toprec}
 &W_{k+1}^g(p,K)=\sum_{\text{branch points }\alpha}\text{Res}_{p\rightarrow \alpha}K_p(q,\bar{q})\notag\\
 &\left(W^{g-1}_{k+2}(q,\bar{q},K)+\sum_{L\cup M=K,h=0}^g W^h_{|L|+1}(q,L)W^{g-h}_{|M|+1}(\bar{q},M)\right)
 \end{align}
 where the sum is restricted to terms with Euler characteristic equal or smaller than 0. For instance if $h=0$ then $|L|\ge 1$. 
 For a very clear exposition about this setup from the point of view of Algebraic Geometry see for instance \cite{MR3087960} but some comments are in order. The branch points are the ones from a meromorphic function $x$ defined on a so called spectral curve $\mathcal{E}(x,y)=0$. The recursion kernel $K_p(q,\bar{q})$ is, roughly speaking, a meromorphic (1,1) tensor that depends on a regular point $p$ in the neighbourhood of a branch point and on $q$ and its conjugated point $\bar{q}$ for which $x(q)=x(\bar{q})$ and $y(q)=-y(\bar{q})$. In fact it can be computed from $W_1^0(p)$ and $W_2^0(p,p)$ which are symmetric differentials of order one and two respectively. Actually, all $W_k^g$ are meromorphic symmetric differentials but we will continue to refer to them as correlation functions. Since our approach will be purely algebraic and in order to soften the notation we will not explicitly mention the sum of the residues over the branch points when referring to this formula.
 \section{The Loday-Ronco Hopf Algebra of planar binary trees}
 We collect here some important facts of the Loday-Ronco Hopf algebra. Details and proofs can be found in \cite{MR2194965,MR1654173}.
 Let $S_n$ be the symmetric group of order $n$ with the usual product $\rho\cdot\sigma$ given by the composition of permutations. When necessary we denote a permutation $\rho$ by its image $(\rho(1)\rho(2)\dots\rho(n))$. Recall that a shuffle $\rho(p,q)$ of type $(p,q)$ in $S_n$ is a permutation such that $\rho(1)<\rho(2)<\dots <\rho(p)$ and $\rho(p+1)<\rho(p+2)<\dots <\rho(p+q)$. For instance the shuffles of type $(1,2)$ in $S_3$ are $(123),(213)$ and $(312)$. We denote the set  of $(p,q)$ shuffles by $S(p,q)$.
 Take 
 \begin{equation}
 k[S^\infty]=\oplus_{n=0}^{\infty}k[S_n]
 \end{equation}
 with $S_0$ identified with the empty permutation. $k[S^\infty]$ is a vector space over a field $k$ of characteristic $0$ generated by linear combinations of permutations. It is graded by the order of permutations and $k[S_0]$ which contains the empty permutation is identified with the field $k$. For two permutations $\rho\in S_p$ and $\sigma\in S_q$ there is a natural product on $S^\infty$ denoted by $\rho\times\sigma$ which is a permutation on $S_{p+q}$ given by letting $\rho$ acting on the first $p$ variables and $\sigma$ acting on the last $q$ variables.

 There is a unique decomposition of any permutation $\sigma\in S_n$ in two permutations $\sigma_i\in S_i$ and $\sigma'_{n-i}\in S_{n-i}$ for each $i$ such that
 \begin{equation}
 \sigma =(\sigma_i\times\sigma'_{n-i})\cdot w^{-1}
 \end{equation}
 where $w$ is a shuffle of type $(i,n-i)$.
 With the $\ast$ product
 \begin{equation}
 \rho\ast \sigma=\sum_{\alpha_{n,m}\in S_{(n,m)}}\alpha_{n,m}\cdot\left(\rho\times\sigma\right)
 \end{equation}
 and the co-product
 \begin{equation}\label{eq:coproductperm}
  \Delta\sigma=\sum \sigma_{i}\otimes\sigma^{'}_{n-i}
  \end{equation}
  $k[S^{\infty}]$ becomes a bi-algebra and since it is graded and connected it is automatically a Hopf Algebra.

 A planar binary tree is a graph with no loops embedded in the plane with only trivalent vertices. In every planar binary tree there are paths that start on a special edge called the root and end on the terminal edges called leaves. The leaves can be left or right oriented. The order $|t|$ of a planar binary tree $t$ is the number of its vertices and on each planar binary tree of order $n$ there are $n+1$ leaves that usually are numbered from 0 to $n$ from left to right. It is frequent to visualize planar binary trees from the bottom to the top, with the root as its lowest vertical edge and the leaves as the highest edges, oriented SW-NE or SE-NW. We will denote the set of planar binary trees of order $n$ by $Y^n$ and by $k[Y^\infty]$ the vector space over $k$ generated by planar binary trees of all orders.
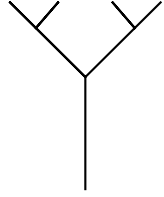
\begin{figure}
   \begin{tikzpicture}
  \draw[thick] (1,-0.5) -- (1,1) -- (0,2) -- (0.35,1.65)--(0.65,2)--(0.35,1.65)-- (1,1) --(1.65,1.65)--(1.35,2)--(1.65,1.65) -- (2,2);
     \end{tikzpicture}
  \caption{A planar binary tree of order 3. Reprinted with permission from \cite{1751-8121-48-44-445205}.}\label{fig:pbt3-2}
\end{figure}
 Additionally a planar binary tree with levels is a planar binary tree such that on each horizontal line there is at most one vertex.
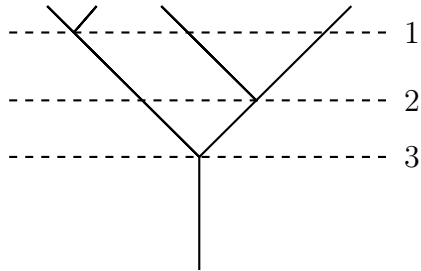
\begin{figure}
   \begin{tikzpicture}
     \draw[thick] (4,-0.5) -- (4,1) -- (2,3) -- (2.35,2.65)--(2.65,3)--(2.35,2.65)-- (4,1) --(4.75,1.75)--(3.5,3)-- (4.75,1.75)--(6,3); \draw[thick,dashed](1.5,1.75)--(6.5,1.75);\draw[thick,dashed](1.5,2.65)--(6.5,2.65);
     \draw (6.8,2.65) node{1}; \draw (6.8,1.75) node{2};\draw[thick,dashed](1.5,1)--(6.5,1);\draw (6.8,1) node{3};
    \end{tikzpicture}
  \caption{Planar binary tree with levels that is the image of $\mathbf{(132)}$. Reprinted with permission from \cite{1751-8121-48-44-445205}.}\label{fig:pbtlev3}
\end{figure}
 It is clear that reading the vertices from left to right and from top to bottom it is possible to assign a permutation of order $n$ to a planar binary tree with levels and that this assignment is unique. For example in fig. \ref{fig:pbtlev3} the tree corresponds to the permutation $(132)$. In this way it is completely equivalent to consider the Hopf algebra $k[S^\infty]$ or the Hopf algebra of planar binary trees with levels because they are isomorphic. However Loday and Ronco show in \cite{MR1654173} that the $\ast$ product and the co-product are internal on the algebra of planar binary trees which is then isomorphic to a sub-Hopf algebra of $k[S^\infty]$ with the same product and co-product. The identity of the Hopf Algebra $k[Y^\infty]$ is the tree with a single edge and no vertices, following the convention of considering only internal vertices, which represents the empty permutation, and the trivial permutation of $S_1$ is represented by the tree with one vertex and two leaves, see fig \ref{fig:idgen}. In fact this element is the generator of the augmented algebra by the $\ast$ product. See fig. (\ref{fig:1star1star1}) for an example of an order 3 product.
 
 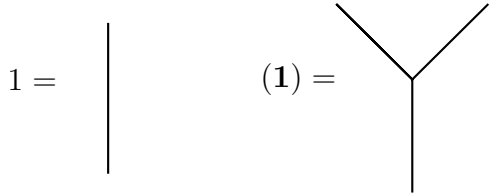
\begin{figure}
  \begin{tikzpicture}
  \draw (0,1) node{ $\Large{1=}$};\draw[thick] (1,-0.25) -- (1,1.75);
  \draw (3.5,1) node{ $\Large{(\mathbf{1})=}$};\draw[thick] (5,-0.5) -- (5,1) -- (4,2) -- (5,1) -- (6,2) ;
  \end{tikzpicture}
  \caption{The identity and the generator in $k[Y^\infty]$. Reprinted with permission from \cite{1751-8121-48-44-445205}.}\label{fig:idgen}
\end{figure}
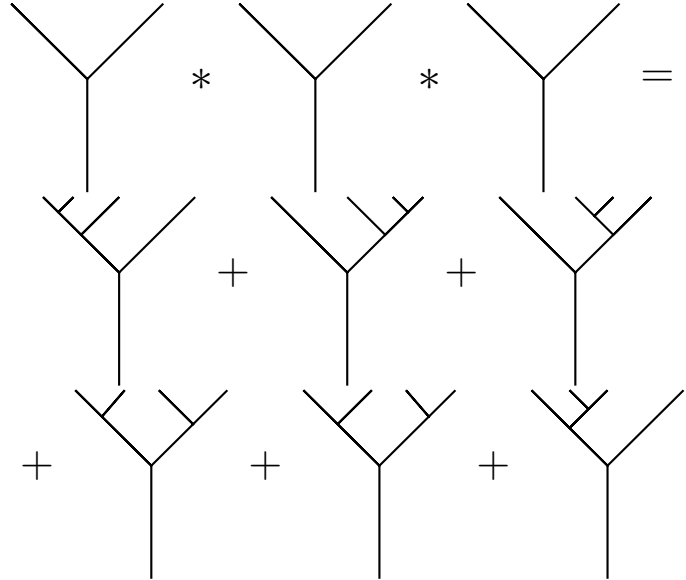
\begin{figure}
  \begin{tikzpicture}
   \draw[thick](1,-0.5) -- (1,1) -- (0,2) -- (1,1) -- (2,2) ; \draw (2.5,1) node{\textbf{{\Large $\ast$}}}; \draw[thick](4,-0.5) -- (4,1) -- (3,2) -- (4,1) -- (5,2); \draw (5.5,1) node{\textbf{{\Large $\ast$}}};\draw[thick](7,-0.5) -- (7,1) -- (6,2) -- (7,1) -- (8,2);\draw (8.5,1) node{\textbf{{\Large $=$}}};
   \end{tikzpicture}
   \begin{tikzpicture}
   \draw[thick] (1,-0.5) -- (1,1) -- (0,2) -- (0.2,1.8)--(0.4,2)--(0.2,1.8)-- (0.5,1.5)--(1,2)--(0.5,1.5)-- (1,1) -- (2,2);\draw (2.5,1) node{\textbf{{\Large $+$}}}; \draw[thick] (4,-0.5) -- (4,1) -- (3,2) -- (4,1) -- (5,2)--(4.8,1.8)--(4.6,2)--(4.8,1.8)--(4.5,1.5)--(4,2);\draw (5.5,1) node{\textbf{{\Large $+$}}}; \draw[thick] (7,-0.5) -- (7,1) -- (6,2) -- (7,1) -- (8,2)--(7.5,1.5)--(7.25,1.75)--(7.5,2)--(7.25,1.75)--(7,2);
    \end{tikzpicture}
    \begin{tikzpicture}
       \draw (-0.5,1) node{\textbf{{\Large $+$}}};\draw[thick] (1,-0.5) -- (1,1) -- (0,2) -- (0.35,1.65)--(0.65,2)--(0.35,1.65)-- (1,1) --(1.55,1.55)--(1.10,2)--(1.55,1.55) -- (2,2);\draw (2.5,1) node{\textbf{{\Large $+$}}};  \draw[thick] (4,-0.5) -- (4,1) -- (3,2) -- (3.45,1.55)--(3.9,2)--(3.45,1.55)-- (4,1) --(4.65,1.65)--(4.35,2)--(4.65,1.65) -- (5,2);\draw (5.5,1) node{\textbf{{\Large $+$}}}; \draw[thick] (7,-0.5) -- (7,1) -- (6,2) -- (6.5,1.5)--(6.75,1.75)--(6.5,2)--(6.75,1.75)--(7,2)-- (6.5,1.5) -- (7,1)-- (8,2);
       \end{tikzpicture} 
\caption{$\mathbf{(1)}\ast\mathbf{(1)}\ast\mathbf{(1)}=\mathbf{(123)}+\mathbf{(321)}+\mathbf{(312)}+\mathbf{(132)}+\mathbf{(231)}+\mathbf{(213)}$ computed in $k[S^\infty]$. Note that in $k[Y^\infty]$ the fourth and the fifth trees are the same. Reprinted with permission from \cite{1751-8121-48-44-445205}.} \label{fig:1star1star1}
 \end{figure}
 The grafting $t_1 \vee t_2$ of two trees $t_1$ and $t_2$ is the operation of producing a new tree $t$ by inserting $t_1$ on the left and $t_2$ on the right leaves of $\mathbf{(1)}$. It is clear that any tree of order $n$ can be written as $t_1\vee t_2$ with $t_1$ of order $p$, $t_2$ of order $q$ and $n=p+q+1$. If a tree has only leaves on the right branch besides the first leaf then it can be written as $|\vee t_2$ and reciprocally if it has only leaves on the left branch besides the last leaf. Note that in particular $\mathbf{(1)}=|\vee |$. 
 
 We can give an explicit expression for the correspondence between a permutation $\sigma\in S^n$ and a planar binay tree $t\in Y^n$. A known notion in combinatorics is the standard permutation associated with a set of different positive integers: say that we have a sequence of $n$ positive integers $a=(a_1,a_2\dots a_n)$ all distinct. The standard permutation $\text{std}(a)$ is the $n-$permutation that has the same ordering of the integers  as the sequence $a$. For instance, for $a=(10,5,6)$, we have that $\text{std}(a)=(312)$. It is also possible to generalize it to sequences of repeated positive integers, but we will not need it. To compute the tree $t$ start by identifiyng the position of $n$ on a permutation $\sigma\in S^n$, $i=\sigma^{-1}(n)$. Then do the standard permutations $\text{std}(\sigma(1)\sigma(2)\dots\sigma(i-1))$ and $\text{std}(\sigma(i+1)\sigma(i+2)\dots\sigma(n))$. Finally the tree $t$ will be grafting of the two trees that correspond to the two standard permutations. Obviously this is a recursive procedure, starting with $|$ for the empty permutation and $\tree$ for the single permutation $(\mathbf{1})$.
 
 To give a concrete example, take the perumutation $(\mathbf{312})$. The left of $3$ is the empty permutation which corresponds to |. The right of $3$ is the permutation $(\mathbf{12})$ which by this iteration procedure can easly be seen to correspond to $\tree\vee |$. In this way we get the third tree on the right of the equation of fig. \ref{fig:1star1star1}.
 
 In \cite{MR1654173} Loday and Ronco show that the $\ast$ product restricted to planar binary trees satisfies the identity
\begin{equation}\label{eq:shuffleident1}
    t\ast t'=t_1\vee (t_2\ast t')+(t \ast t_1^{'})\vee t_2^{'}
\end{equation}
and
\begin{equation}\label{eq:shuffleident2}
    t\ast |=|\ast t=t
\end{equation}
with $t=t_1\vee t_2$ and $t'=t_1^{'}\vee t_2^{'}$.

 If $t_1$ is a tree of order $p$ and a representative element of $W_{p+2}^{0}(p,L)$ and $t_2$ is a tree of order $q$ and a representative element of $W_{q+2}^{0}(p,M)$ then $t=t_1\vee t_2$ is a tree of order $n=p+q+1$ and a representative element of $$W^0_{n+2}(p,K)=K_p(q,\bar{q})W_{p+2}^{0}(q,L)W_{q+2}^{0}(\bar{q},M).$$ with $K=\{p_1,\dots,p_{n+1}\}=L\cup M$.
 We will clarify this in what follows.
 
  \section{The solution of topological recursion\protect\footnote{In this section we will ommit most proofs that can be found in \cite{1751-8121-48-44-445205}.}}
  \subsection{Genus 0}
  The existence of a representation map $\psi$ from the vector space of correlation functions of genus $g$ to the vector space of graphs with loops was suggested in \cite{1751-8121-48-44-445205} such that in particular to a correlation function $W^0_{n+2}(p,p_1,\dots,p_{n+1})$ of Euler characteristic $\chi=-n$ and genus 0 would correspond the trees of order $n$. In fact the way that map was defined implied that the representation of $W^0_{n+2}(p,p_1,\dots,p_{n+1})$ is the sum of all trees of order $n$. In what follows, in particular in section \ref{}, we will see that the map $\psi$ allows to give a ring structure to the space of correlation functions that obey the Eynard-Orantin formula, with an identity that will be given by the cylinder, $W^0_2(p,p_1)$. In fact it is a consequence of the axioms of topological quantum field theory as stated by Atiyah for example in \cite{atiyah1988topological} that the cylinder $\Sigma\times I$, where $\Sigma$ is a topological surface without border and $I$ is a interval of real numbers, may be identified with the identity map in a vector space. 
  \begin{defn}\label{def:W3}
  Consider the planar binary tree with one vertex $(\mathbf{1})$. The 3-point correlation function $W_3^0(p,p_1,p_2)$ is represented by the sum of two planar binary trees with one vertex, obtained by the permutation of the leaf labels $p_1$ and $p_2$:
  \begin{equation}
  \psi\left(W_3^0(p,p_1,p_2)\right)=\sum_{\text{perm. of leaf labels $\{p_1,p_2\}$}}(\mathbf{1})
  \end{equation}
 \end{defn}
 
  The trees that represent $W_3^0(p,p_1,p_2)$ are given by the permutations of the leaf labels of $|\vee |$. Then it is natural to represent the operation of grafting two trees by the insertion of the recursion kernel $K_p(q,\bar{q})$ on its roots. Therefore the symbol
 $\tree $ has two meanings. When isolated it represents $W_3^0(p,p_1,p_2)$ because the two cylinders $W^0_2(q,p_1)$ and $W^0_2(\bar{q},p_2)$ are implicitly identified with its leaves.
 When it is an internal vertex of a more complex tree it is the recursion kernel $K_p(q,\bar{q})$ with suitable labels of its variables.
   \begin{defnbis}{def:W3}\label{def:W3b}
      The propagator or cylinder (also named Bergman kernel in the literature) $W^0_2(q,\bar{q})$ is represented through $\psi$ by the empty permutation $|$ and the recursion kernel is represented through $\psi$ by $\tree$ when in an internal vertex of some tree. Then each planar binary tree of order n is a representation of an instance of some correlation function in genus 0 with each vertex identified with a recursion kernel and each left leaf identified with the cylinder $W^0_2(q_i,p_j)$ or each right leaf identified with the cylinder $W^0_2(\bar{q}_i,p_k)$. Finally the image under $\psi$ of a correlation function $W_{n+2}^0(p,p_1,\dots,p_{n+1})$ with $\chi=-n$ is the sum of all planar binary trees of order $n$ considering all permutations of their leaf labels and with the identifications mentioned above,
      \begin{equation}
      \psi\left(W_{n+2}^0(p,p_1,\dots,p_{n+1})\right)=\sum_{\substack{t_i\in Y^n \\ \text{perm. of leaf labels $\{p_1,\dots,p_{n+1}\}$}}} t_i.
      \end{equation}
    \end{defnbis}
    Hence Definition \ref{def:W3} becomes the following example:
    \begin{defn2}\label{ex:W3}
     Consider the planar binary tree with one vertex. The 3-point correlation function $W_3^0(p,p_1,p_2)$ is represented by the sum of two planar binary trees with one vertex, obtained by the permutation of the leaf labels $p_1$ and $p_2$.
     \end{defn2}
    \begin{align}
     \psi\left(W_3^0(p,p_1,p_2)\right)&=\psi\left(K_p(q,\bar{q})W^0_2(q,p_1)W^0_2(\bar{q},p_2)\right)+ \text{ perm. of $\{p_1,p_2\}$}\notag\\
     &=\sum_{\text{ perm. of $\{p_1,p_2\}$}}|\vee |\notag\\
      &=\sum_{\text{ perm. of $\{p_1,p_2\}$}}(\mathbf{1})
    \end{align}
     
   \begin{prop}
    If $W_{n+2}^0(p,p_1,\dots,p_{n+1})$ is a correlation function with Euler characteristic $\chi=-n$ that is a solution of (\ref{toprec}) then we have
    \begin{align}\label{eq:defcorr}
    \psi \left(W_{n+2}^0(p,p_1,\dots,p_{n+1})\right)&=\sum_{\substack{p+q+1=n\\|t_1|=p, |t_2|=q}} t_1\vee t_2\notag\\
    & + \text{perm. of leaf labels $\{p_1,\dots,p_{n+1}\}$}
    \end{align}
    \end{prop}
 \begin{proof}
 This is the topological recursion in genus $0$ written with planar binary trees. For $n=1$ this is the example \ref{ex:W3}. For $n$ arbitrary by Definition \ref{def:W3b} 
 \begin{equation}
       \psi\left(W_{n+2}^0(p,p_1,\dots,p_{n+1})\right)=\sum_{\substack{t\in Y^n \\ \text{perm. of leaf labels $\{p_1,\dots,p_{n+1}\}$}}} t
       \end{equation}
Decompose uniquely any $t$ of order $n$ into $t=t_1\vee t_2$ of orders $|t_1|=p$ and $|t_2|=q$ with $p+q+1=n$ to get
 \begin{equation}\label{eq:invimage}
       \psi\left(W_{n+2}^0(p,p_1,\dots,p_{n+1})\right)=\sum_{\substack{t_1\in Y^p, t_2\in Y^q, p+q+1=n \\ \text{perm. of leaf labels $\{p_1,\dots,p_{n+1}\}$}}} t_1\vee t_2.
       \end{equation}
Then $t_1$ and $t_2$ are on the image by $\psi$ of $W_{p+2}$ and $W_{q+2}$ for $p$ and $q$ varying from 0 to $n-1$ and constrained by $p+q+1=n$. Since the operation of grafting two trees is represented by attaching the recursion kernel to its roots then, summing for all $t_1\in Y^p, t_2\in Y^q$ and for $p+q+1=n$, we get the topological recursion formula for $g=0$ after taking the preimage of (\ref{eq:invimage}) by $\psi$:
 \begin{align}
 W_{n+2}^0(p,p_1,\dots,p_{n+1})&=\notag\\
 K_p(q,\bar{q})&\sum_{\substack{L\cup M=\{p_1,\dots,p_{n+1}\},\\|L|=p+1,|M|=q+1}} W^0_{|L|+1}(q,L)W^{0}_{|M|+1}(\bar{q},M).
 \end{align} 
 \end{proof}
 \begin{rem}
Note that by $W^0_{n+2}(p,K)$ with $|K|=n+1$ we understand all instances of the correlation function with $g=0$ and $n+2$ labels. This is similar to the situation in High Energy Physics where for the same physical process described by a scattering amplitude there are several Feynman diagrams that contribute.

 It is well known that the dimension of the vector space generated by planar binary trees of order $n$ is given by the Catalan number (see for instance \cite{MR1817703}) $$c_n=\frac{2n!}{n!(n+1)n!}.$$ It is also known that correlation functions in Matrix Models have a large $N$ or planar expansion that is given in terms of Catalan numbers. Therefore it is of no surprise that there exists a correspondence between planar binary trees and correlation functions of genus 0.
 \end{rem}

 \begin{thm}
The $n$-order solution $W_{n+2}^0(p_1,\dots,p_{n+1})$ of the topological recursion in genus 0 is represented by the linear combination
$$\sum t=\mathbf{(1)}\ast\mathbf{(1)}\ast\dots\ast\mathbf{(1)}$$
with $n$ factors of $\mathbf{(1)}$ followed by the sum over all permutations of its labels. 
\end{thm}
In this way $W_{n+2}^0(p,p_1,\dots,p_{n+1})$ is represented by $\sum t$ followed by the identification of cylinders $W^0_2(q_i,p_j)$ with the left leaves or $W^0_2(\bar{q}_i,p_k)$ with the right leaves and finally by summing over all permutations of the labels $p_1,p_2,\dots, p_{n+1}$. In other words, the $\ast$ product $\mathbf{(1)}\ast\mathbf{(1)}\ast\dots\ast\mathbf{(1)}$ gives all possible insertions of recursion kernels of $W_{n+2}^0$.
\subsection{Genus higher than 0}
The procedure of attaching an edge to two consecutive leaves and producing a graph with loops allows to represent correlations functions with genus $g>0$. This is equivalent to extract the outermost cylinders $W^0_2(x,p_j),W_2^0(y,p_{j+1}), x=q_i$ or $\bar{q}_i$, $y=\bar{q}_j$ or $q_j$ and to couple a cylinder $W^0_2(x,y)$  to two recursion kernels $K_{q_l}(q_i,\bar{q}_i)$ and $K_{q_m}(q_j,\bar{q}_j)$, for some convenient choice of indices, that are identified with two internal vertices. This procedure does not change the Euler characteristic of the associated correlation functions because the number of pairs of leaf labels is reduced exactly as the genus is increased. For instance with this procedure we can make the sequence
\begin{equation}
W_5^0(p,p_1,\dots p_4)\longrightarrow W^1_3(p,p_1,p_2)\longrightarrow W^2_1(p)
\end{equation} 
and remain in the same graded vector space that contains $k[Y^3]$. How this changes the Hopf algebra structure is not yet clear. For now, we define the operation $_{i}\leftrightarrow_{i+1}$ on a planar binary tree. 
\begin{defn}\label{defn:connecting}
Starting with a planar binary tree of order $n$ and $n+2$ labels (including the root label $p$) the operation $_{i}\leftrightarrow_{i+1}$ consists in erasing the labels of the leaves $i$ and $i+1$ then connecting them by an edge and finally relabeling the remaining leaves, now numbered $j$ with $j=0,\dots,n-2$, with the $p_{j+1}$ labels, producing in this way a graph with one loop.
\end{defn}
Therefore we represent a correlation function $W^g_{k}(p,p_1,\dots,p_{k-1})$ of genus $g$ by graphs with loops $t^g$ that are obtained by successive applications of the $_{i}\leftrightarrow_{i+1}$ operation. We denote by $\left(Y^n\right)^g$ the set of different graphs with $g$ loops that are obtained from trees $t\in Y^n$.
\begin{defn}
A correlation function $W^g_{k}(p,p_1,\dots,p_{k-1})$ of genus $g$ and Euler characteristic $\chi=2-2g-k$ is represented by a sum of all different graphs with loops $t^g\in \left(Y^n\right)^g$ for  $n=-\chi$:
\begin{equation}
\psi\left(W^g_{k}(p,p_1,\dots,p_{k-1})\right)=\sum_{t^g\in (Y^n)^g} t^g
\end{equation}
\end{defn}
\begin{rem}
Two graphs $(t)^g,(t')^g\in (Y^n)^g$ are considered different in the obvious way. Either the underlying binary trees $t,t'\in Y^n$ are distinct as base elements of $k[Y^\infty]$ or the tree $t$ has a pair of leaves, say $(i,i+1)$ that are identified with an edge in $(t)^g$ producing a loop and are free in $(t')^g$ (and reciprocally because the two graphs have the same genus). 
\end{rem}
In particular $W_1^1(p)$ is represented by a single graph with one loop denoted $(\mathbf{1})^1$ whose underlying planar binary tree is $(\mathbf{1})$. More generally we have
\begin{prop}\label{prop:secterm-W1}
The second summand of the topological recursion formula for the correlation function $W_{n}^1(p,p_1,\dots,p_{n-1})$ with $\chi=-n$ is represented by the sum 
\begin{equation}\label{eq:prop2}
\sum_{\substack{(t_1)^1\in \left(Y^p\right)^1, t_2\in Y^q\\p+q+1=n}}(t_1)^1 \vee t_2 +\sum_{\substack{t_1\in Y^p, (t_2)^1\in \left(Y^q\right)^1\\p+q+1=n}} t_1\vee (t_2)^1
\end{equation}
where each underlying planar binary tree $t$ of order $n$ is decomposed as $t=t_1\vee t_2$, with $|t_1|=p$, $|t_2|=q$, $p+q+1=n$.
\end{prop}
Now we consider the first term in topological recursion which for $W_{n}^1(p,p_1,\dots,p_{n-1})$ is
\begin{equation}\label{eq:firstW^1}
K_p(q,\bar{q})W^0_{n+1}(q,\bar{q},p_1,\dots,p_{n-1}).
\end{equation} 
We start by a definition:
\begin{defn}
The ungrafting operation $\raisebox{1.7mm}{$\line(1,0){10}$}\hspace{-3mm}\vee$ is defined by removing from $t$ the tree $(\mathbf{1})$ that contains the root producing a forest with two trees $t_1$ and $t_2$. When $t$ represents an instance of a correlation function then the roots of $t_1$ and $t_2$ are labeled by $q$ and $\bar{q}$ and as before the tree $(\mathbf{1})$ represents $K_p(q,\bar{q})$.
\end{defn}
\begin{rem}
The operations $\vee$ and $\raisebox{1.7mm}{$\line(1,0){10}$}\hspace{-3mm}\vee$ are similar to the operations $B^+$ and $B^-$ of the Connes-Kreimer Hopf Algebra described, for instance, in \cite{MR1725011}.
\end{rem}
If we start with the planar binary tree $t\in Y^n$ with leaves labels $p_1,\dots, p_{n+1}$ and root label $p$ and identify two nearest neighbor leaves in opposite branches then we get a 1-loop graph $t^1\in \left(Y^n\right)^1$ with a relabeling $p_1,\dots,p_{n-1}$. Then, by applying $\raisebox{1.7mm}{$\line(1,0){10}$}\hspace{-3mm}\vee$, we get another tree $t'$ with two more edges with labels $q$ and $\bar{q}$ besides the leaves labelled by $p_1,\dots,p_{n-1}$. This tree is isomorphic as a graph to a planar binary tree in $Y^{n-1}$ that we denote also $t'$ by promoting the edge with the label $q$ to the root and the other edge to the rightmost leaf, see fig \ref{fig:ungraft} for an example with $W_3^1(p,p_1,p_2)$. In this way we get a representation of (\ref{eq:firstW^1}) by summing over all planar binary trees $t'\in Y^{n-1}$.
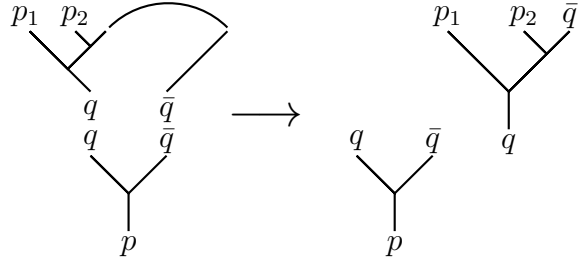
\begin{figure}
 \begin{tikzpicture}
     \draw[thick] (1,-1) -- (1,-0.5) -- (0.5,0) -- (1,-0.5) -- (1.5,0);
     \draw[thick] (0.2,1.15)--(0.7,1.65)--(0.5,1.45)--(0.3,1.65)--(0.5,1.45)--(0.2,1.15)--(-0.3,1.65)-- (0.5,0.85);  \draw[thick] (1.5,0.85)--(1.95,1.3)--(2.3,1.65);
     \draw (0.5,0.2) node{ $q$};\draw (1.5,0.2) node{ $\bar{q}$};
     \draw (0.5,0.6) node{ $q$};\draw (1.5,0.6) node{ $\bar{q}$};
     \draw (1,-1.2) node{ $p$};
      \draw (-0.35,1.85) node{ $p_1$};\draw (0.3,1.85) node{ $p_2$};
      \draw[thick] (2.3,1.7) arc (45:135:1.10cm);
      \draw (2.8,0.5) node{\textbf{{\Large $\longrightarrow$}}};
       \draw[thick] (4.5,-1) -- (4.5,-0.5) -- (4,0) -- (4.5,-0.5) -- (5,0);
       \draw (4.5,-1.2) node{ $p$};
       \draw (4,0.2) node{ $q$};\draw (5,0.2) node{ $\bar{q}$};
        \draw[thick] (6,0.85) -- (5.2,1.65) -- (6,0.85)--(6.8,1.65)--(6.5,1.35) -- (6.2,1.65)--(6.5,1.35)-- (6,0.85)--(6,0.35);
         \draw (6,0.15) node{ $q$};
         \draw (5.2,1.85) node{ $p_1$};\draw (6.2,1.85) node{ $p_2$};
         \draw (6.8,1.85) node{ $\bar{q}$};
      \end{tikzpicture} 
\caption{Ungrafting a 1-loop graph of $W_3^1(p,p_1,p_2)$. The resulting tree is $(\mathbf{21})$. Reprinted with permission from \cite{1751-8121-48-44-445205}.}\label{fig:ungraft}
\end{figure}

In \cite{1751-8121-48-44-445205} we prove the following proposition
\begin{prop}\label{first-term-top-rec}
	The representation of the first term of the topological recursion formula for $W^1_{n}(p,p_1,\dots,p_{n-1})$ is given by the identification of leaves on opposite branches of the decomposition $t=t_1\vee t_2$ with $t\in Y^n, t_1\in Y^p, t_2\in Y^q,p+q+1=n$:
	\begin{equation}\label{eq:prop3}
	\psi\left(K_p(q,\bar{q})W^0_{n+1}(q,\bar{q},p_1,\dots,p_{n-1})\right)=\sum_{\substack{t_1\in Y^p, t_2\in Y^q\\p+q+1=n}}t_1 \bridge t_2 
	\end{equation}
	\end{prop}
The obvious notation $\phantom{.}\bridge\phantom{.}$ means that two consecutive leaves in opposite branches are identified.

Therefore we have exhausted all possibilities of obtaining 1-loop graphs from planar binary trees of order $n$ and the two previous propositions imply the following theorem:
\begin{thm}
The $n$ order solution $W^1_{n}(p,p_1,\dots,p_{n-1})$ of the topological recursion in genus $1$ is given by
$\mathbf{(1)}\ast\mathbf{(1)}\ast\dots\ast\mathbf{(1)}$, with $n$ factors, followed the identification of pairs of nearest neighbor leaves producing 1-loop graphs and finally by  summing over all permutations of leaves labels $p_1,p_2,\dots, p_{n-1}$:
\begin{align}
\psi\left(W^1_{n}(p,p_1,\dots,p_{n-1})\right)&=\notag\\
\sum_{\text{perm. of }\{p_1,\dots,p_{n-1}\}}\sum_{i=0}^{n-1} 
& _{i}\leftrightarrow_{i+1} \left(\mathbf{(1)}\ast\mathbf{(1)}\ast\dots\ast\mathbf{(1)}\right)
\end{align}
\end{thm}

Next we proved in \cite{1751-8121-48-44-445205} a simple lemma regarding symmetric graphs as in fig. \ref{fig:symgraph}. Note that the resulting ungrafted graphs have the left-right order of the right branch of the original graph exchanged:
\begin{lem}
If a graph that enters in the representation of the correlation function $W^{2g+1}_{k+1}(p,K)$ has nearest neighbor leaves identified in different branches and is symmetric with respect to the vertical axis that passes through the root then it has a weight factor of $2$, that is, it appears two times in the complete graph representation of $W^{2g+1}_{k+1}(p,K)$.
\end{lem}
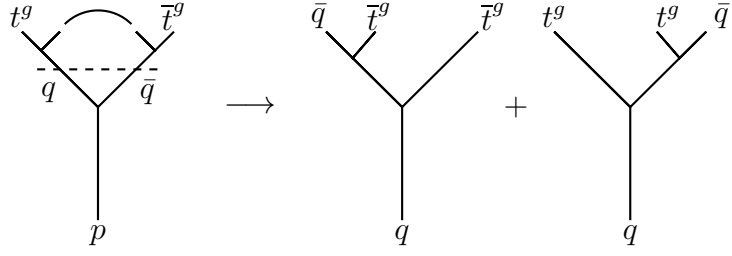
\begin{figure}
 \begin{tikzpicture}
 \draw[thick](1,-0.5) -- (1,1) --(0.25,1.75)--(0.5,2)--(0.25,1.75)-- (0,2) -- (1,1) -- (1.75,1.75)--(1.5,2)--(1.75,1.75)--(2,2) ; \draw (0,2.2) node{ $t^g$};
 \draw (2,2.2) node{ $\overline{t}^g$}; \draw[thick] (1.44,2.1) arc (45:135:0.62cm);
 \draw (3,1) node{ $\longrightarrow$};
 \draw[thick,dashed] (0.2,1.5)--(1.8,1.5);\draw (0.35,1.2) node{ $q$};
 \draw (1.65,1.2) node{ $\bar{q}$};
 \draw (1,-0.7) node{ $p$};
 \draw[thick] (5,-0.5) -- (5,1) -- (4,2) -- (4.35,1.65)--(4.65,2)--(4.35,1.65)-- (5,1) -- (6,2);\draw (6.5,1) node{ $+$}; \draw[thick] (8,-0.5) -- (8,1) -- (7,2) -- (8,1) -- (9,2)--(8.65,1.65)--(8.35,2)--(8.65,1.65);
  \draw (3.9,2.2) node{ $\bar{q}$};\draw (4.7,2.2) node{$\overline{t}^g$};\draw (6.2,2.2) node{ $\overline{t}^g$};
 \draw (7,2.2) node{ $t^g$};\draw (8.5,2.2) node{$t^g$};
  \draw (9.2,2.2) node{ $\bar{q}$};
  \draw (5,-0.7) node{ $q$}; \draw (8,-0.7) node{ $q$};
 \end{tikzpicture} 
\caption{A symmetric $2g+1$-loop graph with $t^g$ a $g-$loop graph. The graph $\overline{t}^g$ is the reflection of $t^g$ on the vertical axis that passes through the root. After the ungrafting operation the first graph is obtained by exchanging $q$ and $\bar{q}$ on the original graph. Reprinted with permission from \cite{1751-8121-48-44-445205}.}\label{fig:symgraph}
\end{figure}
\begin{defn2}
See fig. \ref{fig:W13} for the graph representation of $W_3^1(p,p_1,p_2)$ with a graph of weight 2.
\end{defn2}
\begin{figure}
	\begin{tikzpicture}
	\draw[thick] (1,-0.5) -- (1,1) -- (0,2) -- (0.2,1.8)--(0.4,2)--(0.2,1.8)-- (0.5,1.5)--(1,2)--(0.5,1.5)-- (1,1) -- (2,2);\draw (2.5,1) node{\textbf{{\Large $+$}}}; \draw[thick] (4,-0.5) -- (4,1) -- (3,2) -- (4,1) -- (5,2)--(4.8,1.8)--(4.6,2)--(4.8,1.8)--(4.5,1.5)--(4,2);\draw (5.5,1) node{\textbf{{\Large $+$}}}; \draw[thick] (7,-0.5) -- (7,1) -- (6,2) -- (7,1) -- (8,2)--(7.5,1.5)--(7.25,1.75)--(7.5,2)--(7.25,1.75)--(7,2);\draw[thick] (2,2.1) arc (45:135:0.67cm);\draw[thick] (3.95,2.1) arc (30:150:0.58cm);\draw[thick] (6.95,2.1) arc (30:150:0.58cm);
	\end{tikzpicture}
	\begin{tikzpicture}
	\draw (-0.5,1) node{\textbf{{\Large $+$}}};\draw[thick] (1,-0.5) -- (1,1) -- (0,2) -- (0.35,1.65)--(0.65,2)--(0.35,1.65)-- (1,1) --(1.65,1.65)--(1.35,2)--(1.65,1.65) -- (2,2);\draw (2.5,1);  \draw (3.5,1) node{\textbf{{\Large $+$}}}; \draw[thick] (7,-0.5) -- (7,1) -- (6,2) -- (6.5,1.5)--(6.75,1.75)--(6.5,2)--(6.75,1.75)--(7,2)-- (6.5,1.5) -- (7,1)-- (8,2);
	\draw[thick] (2,2.1) arc (30:150:0.35cm);\draw[thick] (8,2.1) arc (45:135:0.67cm);
	\end{tikzpicture} 
	\begin{tikzpicture}
	\draw[thick] (1,-0.5) -- (1,1) -- (0,2) -- (0.2,1.8)--(0.4,2)--(0.2,1.8)-- (0.5,1.5)--(1,2)--(0.5,1.5)-- (1,1) -- (2,2);\draw (2.5,1) node{\textbf{{\Large $+$}}}; \draw[thick] (4,-0.5) -- (4,1) -- (3,2) -- (4,1) -- (5,2)--(4.8,1.8)--(4.6,2)--(4.8,1.8)--(4.5,1.5)--(4,2);\draw (5.5,1) node{\textbf{{\Large $+$}}}; \draw[thick] (7,-0.5) -- (7,1) -- (6,2) -- (7,1) -- (8,2)--(7.5,1.5)--(7.25,1.75)--(7.5,2)--(7.25,1.75)--(7,2);\draw[thick] (0.4,2.1) arc (30:150:0.2cm);\draw[thick] (5,2.1) arc (45:135:0.3cm);\draw[thick] (8.05,2.1) arc (45:135:0.35cm);
	\end{tikzpicture}
	\begin{tikzpicture}
	\draw (-0.5,1) node{\textbf{{\Large $+$}}};\draw[thick] (1,-0.5) -- (1,1) -- (0,2) -- (0.35,1.65)--(0.65,2)--(0.35,1.65)-- (1,1) --(1.65,1.65)--(1.35,2)--(1.65,1.65) -- (2,2);\draw (2.5,1);  \draw (3.5,1) node{\textbf{{\Large $+$}}}; \draw[thick] (7,-0.5) -- (7,1) -- (6,2) -- (6.5,1.5)--(6.75,1.75)--(6.5,2)--(6.75,1.75)--(7,2)-- (6.5,1.5) -- (7,1)-- (8,2);
	\draw[thick] (0.6,2.1) arc (30:150:0.35cm);\draw[thick] (6.44,2.1) arc (30:150:0.29cm);
	\end{tikzpicture} 
	\begin{tikzpicture}
	\draw[thick] (1,-0.5) -- (1,1) -- (0,2) -- (0.2,1.8)--(0.4,2)--(0.2,1.8)-- (0.5,1.5)--(1,2)--(0.5,1.5)-- (1,1) -- (2,2);\draw (2.5,1) node{\textbf{{\Large $+$}}}; \draw[thick] (4,-0.5) -- (4,1) -- (3,2) -- (4,1) -- (5,2)--(4.8,1.8)--(4.6,2)--(4.8,1.8)--(4.5,1.5)--(4,2);\draw (5.5,1) node{\textbf{{\Large $+$}}}; \draw[thick] (7,-0.5) -- (7,1) -- (6,2) -- (7,1) -- (8,2)--(7.5,1.5)--(7.25,1.75)--(7.5,2)--(7.25,1.75)--(7,2);\draw[thick] (1,2.1) arc (45:135:0.4cm);\draw[thick] (4.6,2.1) arc (45:135:0.4cm);\draw[thick] (7.5,2.1) arc (45:135:0.35cm);
	\end{tikzpicture}
	\begin{tikzpicture}
	\draw (-0.5,1) node{\textbf{{\Large $+2$}}};\draw[thick] (1,-0.5) -- (1,1) -- (0,2) -- (0.35,1.65)--(0.65,2)--(0.35,1.65)-- (1,1) --(1.65,1.65)--(1.35,2)--(1.65,1.65) -- (2,2);\draw (2.5,1);  \draw (3.5,1) node{\textbf{{\Large $+$}}}; \draw[thick] (7,-0.5) -- (7,1) -- (6,2) -- (6.5,1.5)--(6.75,1.75)--(6.5,2)--(6.75,1.75)--(7,2)-- (6.5,1.5) -- (7,1)-- (8,2);
	\draw[thick] (1.27,2.1) arc (45:135:0.42cm);\draw[thick] (7,2.1) arc (30:150:0.29cm);\draw (3.5,-1) node{\textbf{$+$ perm. of $\{p_1,p_2\}$}};
	\end{tikzpicture} 
	\caption{$W_3^1(p,p_1,p_2)$. The root label $p$ and the leaf labels $\{p_1,p_2\}$ are omitted. Reprinted with permission from \cite{1751-8121-48-44-445205}.}\label{fig:W13}
\end{figure}
A simple but important fact is that $_{i}\leftrightarrow_{i+1}$ acts as a derivation when applied independently to the left and right branches of a tree $t=t_1\vee t_2$. This is apparent in Proposition \ref{prop:secterm-W1}. However when summing over all graphs we must take care with overcounting. If we start with $t^1=(t_1)^1\vee t_2+t_1\vee (t_2)^1$ with $|t|=n, |t_1|=p,|t_2|=q, p+q+1=n$ and apply $_{i}\leftrightarrow_{i+1}$ to the two branches independently and sum over all different graphs we get
\begin{align}
&\sum_{i=0}^{n-1} \left(_{i}\leftrightarrow_{i+1}\right)_{\text{same branches}}\left(\sum_{\substack{(t_1)^1\in (Y^{p})^1,t_2\in Y^{q}\\p+q+1=n}}(t_1)^1\vee t_2\right.\notag\\
&\left.+\sum_{\substack{t_1\in Y^p, (t_2)^1\in (Y^{q})^1\\p+q+1=n}}t_1\vee (t_2)^1\right)=\notag\\
&\sum_{\substack{(t_1)^2\in (Y^{p})^2,t_2\in Y^{q}\\p+q+1=n}}2(t_1)^2\vee t_2
+\sum_{\substack{(t_1)^1\in (Y^{p})^1,(t_2)^1\in (Y^{q})^1\\p+q+1=n}}2(t_1)^1\vee (t_2)^1\notag\\
&+\sum_{\substack{t_1\in Y^p, t_2\in Y^{q},(t_2)^2\in (Y^{q})^2\\p+q+1=n}}2t_1\vee (t_2)^2,
\end{align}
whenever the operation is well defined\footnote{The operation is not well defined if there is only one leaf available before and/or after a certain loop or if there are no more leaves to contract. In this case we set to 0 the result of acting with $_{i}\leftrightarrow_{i+1}$.}. The reasoning for the 2 factors is the double counting of identical graphs. For instance if $(t_1)^1$ has a loop starting at leaf $2$ and $(t_1)^2$ was obtained by producing a second loop starting at leaf 0, then this $(t_1)^2$ is identical to the 2-loop graph that was obtained by producing a second loop at leaf 2 in $(t_1)^1$ that had already a loop starting at leaf 0. 

Therefore the sum of different 2-loop graphs obtained from all planar binary trees $t=t_1\vee t_2$ is
\begin{align}
&\sum_{(t)^2\in (Y^n)^2} (t)^2=\sum_{\substack{(t_1)^2\in (Y^{p})^2,t_2\in Y^{q}\\p+q+1=n}}(t_1)^2\vee t_2\notag
\end{align}
\begin{align}
 &+\sum_{\substack{(t_1)^1\in (Y^{p})^1,(t_2)^1\in (Y^{q})^1\\p+q+1=n}}(t_1)^1\vee (t_2)^1\notag\\
 &+\sum_{\substack{t_1\in Y^p, (t_2)^2\in (Y^{q})^2\\p+q+1=n}}t_1\vee (t_2)^2.
\end{align}

As for the identification of two consecutive leafs in separate branches, we start again from $t^1=(t_1)^1\vee t_2+t_1\vee (t_2)^1$ to get
\begin{equation}
t^2=(t_1)^1\bridge t_2+t_1\bridge (t_2)^1
\end{equation}
The fact that it may not be always possible to contract two leaves in opposite branches is important to count the dimensions of the vector spaces $k[(Y^n)^g]$ generated by graphs with $g$ loops but we will not consider this. 

It is clear that we can continue this procedure and generate graphs $t^g\in (Y^n)^g$ with an increasing number of $g$ loops and $k$ labels (including the root) up to the consistence of the relation $-n=2-2g-k$.
Hence we have the following proposition:
\begin{prop}\label{prop:graphloop}
The graphs with loops $t^g\in (Y^n)^g$ and $k$ labels, including the root, that are compatible with $-n=2-2g-k$ are obtained from planar binary trees $t=t_1\vee t_2$ by successive applications of $\left(_{i}\leftrightarrow_{i+1}\right)_{\text{same branches}}$ and $\left(_{i}\leftrightarrow_{i+1}\right)_{\text{opposite branches}}$ to all $t\in Y^n$:
\begin{align}\label{eq:gloopgraph}
\sum_{(t)^g\in (Y^n)^g} (t)^g=&\sum_{k=0}^g\sum_{\substack{(t_1)^k\in(Y^p)^k,(t_2)^{g-k}\in(Y^q)^{g-k}\\p+q+1=n}} \left( (t_1)^k\vee (t_2)^{g-k}\right)\notag\\
&+\sum_{k=0}^{g-1} \sum_{\substack{(t_1)^k\in(Y^p)^k,(t_2)^{g-k}\in(Y^q)^{g-k}\\p+q+1=n}}\left((t_1)^{g-1-k}\bridge (t_2)^{k}\right)
\end{align}
\end{prop}

\begin{thm}
The $n$ order solution $W^g_{2-2g+n}$ of the topological recursion in genus $g>0$ and $k=2-2g+n>0$ variables is given by
$\mathbf{(1)}\ast\mathbf{(1)}\ast\dots\ast\mathbf{(1)}$, with $n$ factors, followed the identification of pairs of nearest neighbor leaves producing graphs with loops as in Proposition \ref{prop:graphloop} and finally by  summing over all permutations of $p_1,p_2,\dots, p_{1-2g+n}$.
\end{thm}
\begin{rem}
After being ungrafted the second term in (\ref{eq:gloopgraph}) can still have leaves in opposite branches identified. In general this happens if $(t)^g=(t_1)^{g_1}\bridge (t_2)^{g_2}$ with $g_1+g_2=g-1$ and say $(t_1)^{g_1}$ has a decomposition $(t_1)^{g'_1}\bridge (t_2)^{g'_2}$ with ${g'}_1+{g'}_2=g_1-1$.
 \begin{defn2}
 See fig. \ref{fig:W21} for the graph representation of $W_1^2(p)$.
 \end{defn2}
  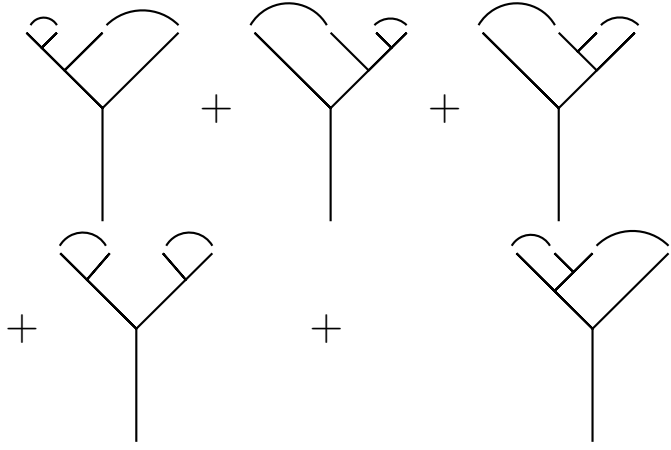
\begin{figure}[h]
 	\begin{tikzpicture}
 	\draw[thick] (1,-0.5) -- (1,1) -- (0,2) -- (0.2,1.8)--(0.4,2)--(0.2,1.8)-- (0.5,1.5)--(1,2)--(0.5,1.5)-- (1,1) -- (2,2);\draw (2.5,1) node{\textbf{{\Large $+$}}}; \draw[thick] (4,-0.5) -- (4,1) -- (3,2) -- (4,1) -- (5,2)--(4.8,1.8)--(4.6,2)--(4.8,1.8)--(4.5,1.5)--(4,2);\draw (5.5,1) node{\textbf{{\Large $+$}}}; \draw[thick] (7,-0.5) -- (7,1) -- (6,2) -- (7,1) -- (8,2)--(7.5,1.5)--(7.25,1.75)--(7.5,2)--(7.25,1.75)--(7,2);\draw[thick] (0.4,2.1) arc (30:150:0.2cm);\draw[thick] (2,2.1) arc (45:135:0.67cm);\draw[thick] (3.95,2.1) arc (30:150:0.58cm);\draw[thick] (5,2.1) arc (45:135:0.3cm);\draw[thick] (6.95,2.1) arc (30:150:0.58cm);\draw[thick] (8.05,2.1) arc (45:135:0.35cm);
 	\end{tikzpicture}
 	\begin{tikzpicture}
 	\draw (-0.5,1) node{\textbf{{\Large $+$}}};\draw[thick] (1,-0.5) -- (1,1) -- (0,2) -- (0.35,1.65)--(0.65,2)--(0.35,1.65)-- (1,1) --(1.65,1.65)--(1.35,2)--(1.65,1.65) -- (2,2);\draw (2.5,1);  \draw (3.5,1) node{\textbf{{\Large $+$}}}; \draw[thick] (7,-0.5) -- (7,1) -- (6,2) -- (6.5,1.5)--(6.75,1.75)--(6.5,2)--(6.75,1.75)--(7,2)-- (6.5,1.5) -- (7,1)-- (8,2);
 	\draw[thick] (0.6,2.1) arc (30:150:0.35cm);\draw[thick] (2,2.1) arc (30:150:0.35cm);\draw[thick] (6.44,2.1) arc (30:150:0.29cm);\draw[thick] (8,2.1) arc (45:135:0.67cm);
 	\end{tikzpicture} 
 	\caption{$W_1^2(p)$ The root label $p$ is omitted. Reprinted with permission from \cite{1751-8121-48-44-445205}.}\label{fig:W21}
 \end{figure}
\end{rem}
\section{The antipode}
In a graded connected Hopf Algebra there is a canonical antipode $S$ whose expression is given by the convolution inverse of the identity:
\begin{equation}
m\left(S\otimes \text{Id}\right)\Delta = m\left(I\otimes \text{S}\right)\Delta=\eta\cdot\epsilon
\end{equation}
with $m$ the product, $\Delta$ the co-product, $\eta$ the unit and $\epsilon$ the co-unit. Explicitly, in the Loday-Ronco Hopf Algebra, we have
\begin{equation}
S(t)=-t-S(t_1)\ast t_2
\end{equation}
where in Sweedler notation
\begin{equation}
\Delta t = \sum t_1\otimes t_2
\end{equation}
is the co-product in $k[Y^\infty]$ induced by (\ref{eq:coproductperm}).
For instance, $S(\mathbf{1})=-\mathbf{1}$ because $\mathbf{1}$ is primitive and $S(\mathbf{12})=(\mathbf{21})$ and also $S(\mathbf{21})=(\mathbf{12})$. This suggests that a map $\psi^\ast S$ induced by the antipode on the vector space of correlations functions should give 
$$(\psi^\ast S)(W^0_4)=W^0_4.$$
More generally, from $S((\mathbf{1}))=-(\mathbf{1})$ we see that
$$S((\mathbf{1})\ast(\mathbf{1})\ast\dots\ast(\mathbf{1}))=(-1)^n(\mathbf{1})\ast(\mathbf{1})\ast\dots\ast(\mathbf{1})$$
with $n$ factors in the $\ast$ product and then $$(\psi^\ast S)(W_{n+2}^0)=(-1)^nW_{n+2}^0.$$
Since $n$ is identified with the Euler characteristic we see that the induced map respects the grading of $k[Y^\infty]$ for $g=0$.

\section{Discussion and other topics}
This section contains some new material which may not be in a complete final form. 

We would like to give a more direct approach to Topological Recursion in genus 0 in terms of the co-product of the Loday-Ronco Hopf algebra and analyze the algebraic structure of the space of correlation functions. We start by recalling that the elementary tree $\tree$ is the same as the elementary permutation $(\mathbf{1})$ and that this tree is primitive in the Loday-Ronco Hopf algebra: 
\begin{equation}
\Delta\tree=1\otimes\tree+\tree\otimes 1.
\end{equation}
Now by Theorem 1 a solution of Topological Recursion of Euler characteristic $-n$ in genus 0 is given by the product of $n$ terms of $(\mathbf{1})$, modulo the permutations of leaves labels, and since the co-product is a homomorphism we have that
\begin{align}
\Delta \left((\mathbf{1})^n\right)&=\left(\Delta (\mathbf{1})\right)^n\notag\\
&=\left(1\otimes(\mathbf{1})+(\mathbf{1})\otimes 1\right)^n\notag\\
&=\sum_{k=0}^n C^n_k (\mathbf{1})^k\otimes(\mathbf{1})^{n-k}
\end{align}
where in the powers the Loday-Ronco $\ast$ product is understood. If we absorb the factor $n!$ in a normalization of $n$ factors of $(\mathbf{1})$ then the co-product is
\begin{equation}\label{eq:co-productpower}
\Delta \left((\mathbf{1})^n\right)=\sum_{k=0}^n (\mathbf{1})^k\otimes(\mathbf{1})^{n-k}
\end{equation}
This is exactly the same situation that appears in the elementary Hopf algebra $k[X]$ of polynomials in one variable $X$ over some field $k$ of characteristic 0. By declaring the variable X to be primitive and normalizing $X^n$ with $n!$ we get
\begin{equation}
\Delta \left(\frac{X^n}{n!}\right)=\sum_{k=0}^n\frac{X^k}{k!}\otimes \frac{X^{n-k}}{(n-k)!}
\end{equation}
With the map $\psi$ we can translate this to the language of correlation functions in genus zero $W^0_{n+2}(p,p_1,\dots,p_{n+1})$ with Euler characteristic $\chi=-n$: 
\begin{align}\label{co-productcorr}
\Delta \left(W^0_{n+2}(p,p_1,\dots,p_{n+1})\right)&=\notag\\
\sum_{k=0}^n W^0_{k+2}(p,p_1,\dots,p_k,\bar{q})&\otimes W^0_{n-k+2}(\bar{q},p_{k+1},\dots,p_{n+1})\notag\\
+\sum_{k=0}^n &W^0_{k+2}(q,p_{1},\dots,p_{k+1})\otimes W^0_{n-k+2}(p,q,p_{k+2},\dots,p_{n+1})
\end{align}
where a normalization of $n!$ for $W^0_{n+2}$ is implicit. An explanation for the extra labels $q,\bar{q}$ and also for the two sums is in order. Note that the co-product in (\ref{eq:co-productpower}) just amounts to cut the product of $n$ $(\mathbf{1})$'s into two consecutive factors. At the level of trees this amounts to cut a complex tree into two more elementary trees, where the cut is done along an internal vertex, giving a tree that contains the original root and a new leaf just below the cut and another tree that contains a new root just above the cut and some leaves of the original tree. In this way, every time a tree is cut along an internal vertex we assign the label $q$ if the new leaf is left oriented and $\bar{q}$ if it is right oriented. Moreover the new root of the second tree is also labeled $q$ and $\bar{q}$ respectively. 
\begin{figure}
	\begin{tikzpicture}[scale=1, transform shape]
	\draw (0,0)--(0,1)--(-0.5,1.5)--(0,1)--(0.5,1.5);\draw (0,-0.3) node{ $p$};\draw (1.65,1.85) node{ $W^0_{k+1}(\bar q,p_1,\dots,p_k)$};\draw (-0.5,1.8) node{ $q$};
	\end{tikzpicture}\caption{$W^0_{k+2}(p,q,p_1,\dots,p_k)$}\label{fig:W0kpq}
\end{figure}
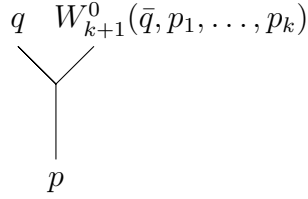
That is why it is necessary to take into account the two decompositions, say a left and a right one. Hence in our approach the object $W^0_{k+2}(p,q,p_1,\dots,p_k)$ that contains simultaneously the root $p$ and an internal leaf label $q_i$ (or $\bar q_i$) is an intermediate object between a correlation function and a recursion kernel in the sense that the internal label does not couple to a cylinder, see fig. \ref{fig:W0kpq} 
As an example take for instance $W^0_4(p,p_1,p_2,p_3)$. We get
\begin{align}\label{coprodcutW04}
\Delta W^0_4(p,p_1,p_2,p_3)&=1\otimes W^0_4(p,p_1,p_2,p_3)+W^0_4(p,p_1,p_2,p_3)\otimes 1+\notag\\
 W^0_3(p,p_1,\bar{q})\otimes & W^0_3(\bar{q},p_2,p_3)+W^0_3(q,p_1,p_2)\otimes W^0_3(p,q,p_3).
\end{align}
\begin{figure}
	\begin{tikzpicture}[scale=2, transform shape]
	\draw[thick] (6.25,0.75) -- (6.25,1) -- (6,1.25) -- (6.125,1.125)--(6.25,1.25)--(6.125,1.125) -- (6.25,1) -- (6.5,1.25);
	\draw[thick] (7,0.75) -- (7,1) -- (6.75,1.25) -- (7,1)  -- (7.125,1.125)--(7,1.25)--(7.125,1.125)--(7.25,1.25); \draw (6.625,1) node{ +};
	\end{tikzpicture}
	\caption{Representation of $W^0_4(p,p_1,p_2,p_3)$. Root and leaves labels are not shown.}
	\label{fig:W04}
\end{figure}
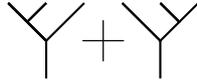
If we look at the two trees in fig. \ref{fig:W04} we see that only the tree $(\mathbf{21})$ admits a decomposition like the third term and only the tree $(\mathbf{12})$ admits a decomposition like the fourth term.
Next translating $W^0_3(p,p_1,q)$ into the expression
\begin{equation}
W^0_3(p,p_1,q)=K_p(q,\bar{q})W^0_2(\bar{q},p_1)
\end{equation}
we get the terms of Topological Recursion for $W^0_4(p,p_1,p_2,p_3)$. Note that $W^0_3(p,p_1,q)$ carries an internal index that results from an internal cut, so this index does not couple to a cylinder.

Now is easy to see that on one hand the expression (\ref{co-productcorr}) gives all the terms of the first iteration of Topological Recursion in genus 0 for $W^0_{n+2}$ and on the other hand that by successive iterations of this co-product we end up with the complete decomposition of the solution in terms of recursion kernels. More specifically what we need is the reduced co-product $\Delta'$ that is equal to the co-product minus the primitive part:
\begin{equation}
\Delta' a=\Delta a -1\otimes a-a\otimes 1
\end{equation}
for a generic element $a$ in some Hopf algebra. Note that $\Delta'(a)=0$ for a primitive element $a$, and also by definition $\Delta'1=0$. In particular,
\begin{equation}
\Delta'W^0_3(p,p_1,p_2)=0.
\end{equation}
Then we take the iteration $\Delta'^{(n)}$:
\begin{align}
\Delta'^{(1)}&=\Delta'\notag\\
\Delta'^{(2)}&=(\Delta'\otimes\text{Id})\otimes \Delta'^{(1)}\notag\\
\Delta'^{(n)}&=(\Delta'\otimes\text{Id})\Delta'^{(n-1)}.
\end{align}
Let us take a product of $n$ terms of $(\mathbf{1})$ and apply $\Delta'^{(n-1)}$:
\begin{align}
\Delta'\left((\mathbf{1})^n\right)&=\left(1\otimes(\mathbf{1})+(\mathbf{1})\otimes 1\right)^n-(\mathbf{1})^n\otimes 1-1\otimes(\mathbf{1})^n\notag\\
&=\sum_{k=1}^{n-1}(\mathbf{1})^k\otimes (\mathbf{1})^{n-k};\notag\\
\Delta'^{(2)}\left((\mathbf{1})^n\right)&=(\Delta'\otimes\text{Id})\left(\sum_{k_1=1}^{n-1}(\mathbf{1})^{k_1}\otimes (\mathbf{1})^{n-k_1}\right)\notag\\
&=\left(\sum_{k_1=2}^{n-1}\sum_{k_2=1}^{k_1-1}(\mathbf{1})^{k_2}\otimes (\mathbf{1})^{k_1-k_2}\otimes(\mathbf{1})^{n-k_1}\right)\notag\\
\Delta'^{(n-1)}\left((\mathbf{1})^n\right)&=(\mathbf{1})\otimes(\mathbf{1})\otimes\dots\otimes(\mathbf{1}), n\text{ factors}.
\end{align}
where we have absorbed the combinatorial factors on a normalization of $(\mathbf{1})$ factors. However the simplicity of this expression is misleading because we are ignoring the root and leaves labels. In fact all $(\mathbf{1})$'s represent recursion kernels that carry coupled internal indexes.

Now that we have discussed the co-product at the level of correlation functions of genus 0 and not at the level of planar binary trees and have given the decomposition of Topological Recursion in terms of recursion kernels using this co-product, we'd better see if there is a product in the first place. But this is immediately induced by our solution stated by Theorem 1. Given two correlation functions $W^0_{k+2}$ and $W^0_{l+2}$ (here we let fall the map $\psi$ which we implicitly assume)
\begin{align}
&W^0_{k+2}(p_1,\dots,p_{k+1})=(\mathbf{1})\ast(\mathbf{1})\ast\dots\ast(\mathbf{1}), k\text{ factors},\notag\\
&W^0_{l+2}(p_1,\dots,p_{l+1})=(\mathbf{1})\ast(\mathbf{1})\ast\dots\ast(\mathbf{1}), l\text{ factors},
\end{align}
then their product is
\begin{align}\label{eq:productcorr}
&W^0_{k+2}(p,p_1,\dots,p_{k+1})\ast W^0_{l+2}(p',p'_1,\dots,p'_{l+1})\notag\\ &=W^0_{m+2}(p,p_1,\dots,p_{m+1})=(\mathbf{1})\ast(\mathbf{1})\ast\dots\ast(\mathbf{1}), m=k+l\text{ factors}
\end{align}
which is clearly associative, commutative and with the unit $W^0_2(p,p_1)$. However we are ignoring the role of the root and leaves labels. It is possible to generalize this product in order to take into account the labels. First we consider the product of a cylinder $W^0_2(p,p_1)$ and a generic correlation function $W^0_{n+2}(p',p'_1,\dots,p'_{n+1})$:
\begin{align}
W^0_2(p,p_1)\ast W^0_{n+2}(p',p'_1,\dots,p'_{n+1})=&W^0_{n+2}(p',p'_1,\dots,p'_{n+1})\ast W^0_2(p,p_1)\notag\\
=& W^0_{n+2}(p,p'_1,\dots,p'_{n+1})
\end{align}
This amounts to make $p'=p_1$ and replace the root label of $W^0_{n+2}$ by $p$. Next we use the expression (\ref{eq:shuffleident1}) for the product of two planar binary trees in order to get the product of two $W^0_3$:
\begin{align}
&W^0_3(p,p_1,p_2)\ast W^0_3(p',p'_1,p'_2)=\notag\\
&W^0_2(q,p_1)\vee (W^0_2(\bar q,p_2)\ast W^0_3(p',p'_1,p'_2)\notag\\
&+(W^0_3(p,p_1,p_2)\ast W^0_2(q',p'_1))\vee W^0_2(\bar q',p'_2)\notag\\
&=W^0_2(q,p_1)\vee W^0_3(\bar q,p'_1,p'_2)+W^0_3(q',p_1,p_2)\vee W^0_2(\bar q',p'_2)\notag\\
&=K_p(q,\bar q)W^0_2(q,p_1) W^0_3(\bar q,p'_1,p'_2)+K_{p'}(q',\bar q')W^0_3(q',p_1,p_2)W^0_2(\bar q',p'_2),
\end{align}
where we have used the decomposition 
\begin{equation}
W^0_3(p,p_1,p_2)=K_p(q,\bar q)W^0_2(q,p_1)W^0_2(\bar q,p_2)
\end{equation} 
In the particular case of equal root labels and $p'_1=p_2,p'_2=p_3$, we get
\begin{equation}
W^0_3(p,p_1,p_2)\ast W^0_3(p,p_2,p_3)=W^0_4(p,p_1,p_2,p_3).
\end{equation}
Note that this product is no longer commutative and is additive on the Euler characteristic. Finally we define in a similar way the product for two generic correlation functions of $g=0$:

\begin{align}
&W^0_{l+2}(p,p_1,\dots,p_{l+1})\ast W^0_{m+2}(p',p^{'}_{1}\dots,p^{'}_{m+1})=\notag\\
&\delta_{p,p'}\delta_{p_{l+1},p'_{1}}W^0_{k+2}(p,p_1,\dots,p_{k+1}),k=l+m.
\end{align}
 with the redefinition $p_{l+2}=p^{'}_{2},\dots,p_{l+m+1}=p^{'}_{m+1}$. We show that this product is compatible with Topological Recursion. Start with
 \begin{align}
 &W^0_{l+2}(p,p_1,\dots,p_{l+1})=\notag\\
 &K_p(q,\bar q)\sum_{k=0}^{l-1}W^0_{k+2}(q,\dots,p_{k+1})W^0_{l-k+1}(\bar q,p_{k+2},\dots,p_{l+1})\notag\\
 &=\underset{t_1\in Y^p,t_2\in Y^q, p+q+1=l}{\sum} t_1\vee t_2
\end{align}
and similarly for $W^0_{m+2}(p,p_1,\dots,p_{m+1})$. In terms of planar trees and using (\ref{eq:shuffleident1}) we get
 \begin{align}
 &W^0_{l+2}(p,p_1,\dots,p_{l+1})\ast W^0_{m+2}(p',p^{'}_{1}\dots,p^{'}_{m+1})=\notag\\
 &\underset{t_1\in Y^p,t_2\in Y^q, p+q+1=l}{\sum} \underset{t^{'}_1\in Y^{p'},t^{'}_2\in Y^{q'}, p'+q'+1=m}{\sum}t_1\vee (t_2\ast t')+(t\ast t^{'}_1)\vee t^{'}_2.
 \end{align}
 Translating this into correlation functions,
 \begin{align}
  &W^0_{l+2}(p,p_1,\dots,p_{l+1})\ast W^0_{m+2}(p',p^{'}_{1}\dots,p^{'}_{m+1})=K_p(q,\bar q)\times\notag\\
 &\left(\sum_{k=0}^{l-1}W^0_{k+2}(q,p_1,\dots,p_{k+1})\vee\right.\notag\\
 &\left.\left(W^0_{l-k+1}(\bar q,p_{k+2},\dots,p_{l+1})\ast W^0_{m+2}(p',p^{'}_{1}\dots,p^{'}_{m+1})\right)\right)\notag\\
 &+K_{p'}(q',\bar q')\times\notag\\
 &\left(\sum_{n=0}^{m-1}\left(W^0_{l+2}(p,p_1,\dots,p_{l+1})\ast W^0_{n+2}(q',p^{'}_1,\dots,p^{'}_{n+1})\right)\vee \right.\notag\\
 &\left.W^0_{m-n+1}(\bar q',p^{'}_{n+2},\dots,p^{'}_{m+1})\right)\notag\\
 &=K_p(q,\bar q)\times\notag\\
 &\left(\sum_{k=0}^{l-1}W^0_{k+2}(q,p_1,\dots,p_{k+1})W^0_{l+m-k+1}(\bar q,p_{k+2},\dots,p_{l+m+1})\right.\notag\\
 &\left.+\sum_{n=0}^{m-1}W^0_{l+n+2}(q,p_1,\dots,p_{l+n+1})W^0_{m-n+1}(\bar q,p^{'}_{n+2},\dots,p^{'}_{m+1})\right)\notag\\
 &=K_p(q,\bar q)\times\notag\\
 &\left(\sum_{k=0}^{l-1}W^0_{k+2}(q,p_1.\dots,p_{k+1})W^0_{l+m-k+1}(\bar q,p_{k+2},\dots,p_{l+m+1})\right.\notag\\
 &\left.+\sum_{k=l}^{m+l-1}W^0_{k+2}(q,p_1,\dots,p_{k+1})W^0_{l+m-k+1}(\bar q,p^{'}_{k-l+2},\dots,p^{'}_{m+1})\right)\notag\\
 &=W^0_{l+m+2}(p,p_{1},\dots,p_{l+m+1})
 \end{align}
 for $p=p'$ and with obvious redefinitions of the variables.
 Therefore we recover (\ref{eq:productcorr}) and confirm the fact that the space of correlation functions has an algebra structure. It is also a bi-algebra as the product and the co-product are induced by these same structures in the Loday-Ronco Hopf algebra. The antipode has been discussed in the previous section so this algebra is indeed a Hopf algebra.
\subsection{The Connes-Kreimer Hopf algebra}
Here we won't have the opportunity to give a detailed exposition of the intricacies of the Renormalization procedure in Quantum Field Theory \cite{collins1984renormalization} and how the combinatorics of this procedure are well described by Connes-Kreimer Hopf algebra. A good review can be found in \cite{MR1725011}. We just mention that the space of this Hopf algebra is the vector space generated by rooted trees over $\Rational$, that are graphs with no loops, with a preferred vertex called the root that has only child edges, and with vertices's of any order. These trees are not planar and the product in the algebra is just disjoint union. For instance, the elementary tree $\bullet$ describes any Feynman diagram that is divergent but that has no subdivergent diagrams, that are subdiagrams of a more complex diagram that are by themselves divergent and whose divergent integral is a function of the integration variable of the diagram to which they are a part of\footnote{We should say in passing that a regularization procedure is necessary to give a meaning to the manipulation of these integrals and this is tacitly assumed.}. This situation goes on order-by-order in perturbation theory, or equivalently in loop order, giving more and more nested subdiagrams and subdivergences. The first example of a nested divergence is the tree $\treetwo$ that describes any divergent Feynman diagram with a subdivergent diagram that doesn't have by itself any subdiagram. Any rooted tree of any order with an arbitrary number of branches can be used to describe some complicated divergent integrals associated to some Feynman diagrams. See for instance \cite{Brouder:1999gk,Brouder:1999za,BROUDER2003298}. Note that a particularly interesting aspect is the fact that rooted trees in the Connes-Kreimer algebra describe what are known as One Particle Irreducble (OPI) Feynman diagrams. These are connected diagrams that represent the Legendre transform of connected correlation functions. Its is well known that in Quantum Field Theory connected correlation functions are generated by the free energy $F$ while products of connected correlation functions are generated by the partition function $Z$: 
\begin{equation}
Z=\exp F.
\end{equation}
We will come back to this later.

There is a generalization due to Foissy \cite{MR1909461} of the Connes-Kreimer Hopf algebra called the noncommutative Connes-Kreimer Hopf algebra that considers planar rooted trees. The product is ordered disjoint union and is no longer commutative as for instance one distinguishes $\bullet \treetwo$ from $\treetwo \bullet$.
The article \cite{MR2194965} describes a nice inductive map $\phi$ to compute the planar binary tree that corresponds to a planar rooted tree or a planar forest in the noncommutative Connes-Kreimer Hopf algebra. We need to introduce two operations $\backslash$ and $/$ on trees $t_1$ and $t_2$, $t_1\backslash t_2$ and $t_1/t_2$. The first amounts to graft $t_2$ on the right most leaf of $t_1$ and the second is the grafting of $t_1$ on the left most leaf of $t_2$, see fig. \ref{fig:backslash} for a simple example (there are corresponding operations on permutations that can also be found in \cite{MR2194965}).
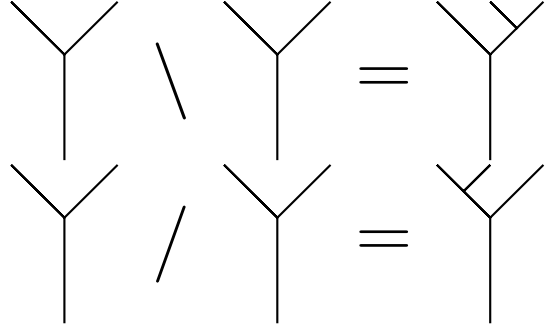
\begin{figure}
	\begin{tikzpicture}[scale=1.4, transform shape]
	\draw[thick](1,0) -- (1,1) -- (0.5,1.5) -- (1,1) -- (1.5,1.5) ; 
	\draw (2,.75) node{ \LARGE  $\backslash$};
	\draw[thick](3,0) -- (3,1) -- (2.5,1.5) -- (3,1) -- (3.5,1.5) ; 
	\draw (4,.75) node{ \LARGE  $=$};
	\draw[thick](5,0) -- (5,1) -- (4.5,1.5) -- (5,1) -- (5.25,1.25)--(5,1.5)--(5.25,1.25) -- (5.5,1.5) ; 
	\end{tikzpicture}
	\begin{tikzpicture}[scale=1.4, transform shape]
	\draw[thick](1,0) -- (1,1) -- (0.5,1.5) -- (1,1) -- (1.5,1.5) ; 
	\draw (2,.75) node{ \LARGE  $/$};
	\draw[thick](3,0) -- (3,1) -- (2.5,1.5) -- (3,1) -- (3.5,1.5) ; 
	\draw (4,.75) node{ \LARGE  $=$};
	\draw[thick](5,0) -- (5,1) --(4.75,1.25)--(5,1.5) -- (4.75,1.25)--(4.5,1.5) -- (5,1) -- (5.25,1.25)-- (5.5,1.5) ; 
	\end{tikzpicture}
\caption{The $\backslash$ and $/$ operation on trees.}\label{fig:backslash}
\end{figure}
First, $\phi(\emptyset)=|$. Next, removing the root from a planar rooted tree $t$ gives a planar forest $f=(t_1,t_2,\dots,t_n)$ with each $t_i$ a planar rooted tree and we set $\phi(t)=\phi(f)/\tree$. Finally for a forest $f$ we do
$\phi(f)=\phi(t_1)\backslash\phi(t_2)\backslash...\backslash\phi(t_n)$.
For instance, $\phi(\bullet)=\tree$, $\phi(\treetwo)=\phi(\bullet)/\phi(\bullet)=(\mathbf{12})$, $\phi(\bullet\bullet)=\phi(\bullet)\backslash\phi(\bullet)=(\mathbf{21})$.

Now exponentiate the elementary tree $\mathbf{(1)}$ multiplied by a constant $g$ that we would like to see as a coupling constant in some physical theory:
\begin{align}
\exp(g\tree)&=1+g\tree+\frac{g^2}{2}\tree\ast\tree\notag\\
&+\frac{g^3}{3!}\tree\ast\tree\ast\tree+\dots\notag\\
&\underset{\phi^{-1}}{\longrightarrow}1+g\bullet +\frac{g^2}{2}\left(\treetwo+\bullet\bullet\right)\notag\\
&+\frac{g^3}{3!}\left(\treethree+\treethreetwo+\treetwo\bullet+\bullet\treetwo+\bullet\bullet\bullet\right)+\dots
\end{align}
We see that the exponential series of $\tree$ generates at some specific order in $g$ all possible divergences in loops of the Connes-Kreimer Hopf algebra. Remember that this algebra is in correspondence with nested divergences of the OPI diagrams, that are connected and have a definite loop order, and so is directly related with the free energy $F$. Also the elementary tree $\bullet$ may represent any OPI diagram of any loop order that does not have subdivergences. On the other hand the Loday-Ronco algebra seems to generate all divergences in all possible loop orders compatible with a specific order in the coupling constant $g$. Hence it appears to be related with the expansion of the partition function $Z$ in terms of (not connected) correlation functions of some physical theory. We would like to clarify this in future work.
\subsection{Genus higher than 0}
We start with the space of correlation functions of genus 0 that are solutions of Topological Recursion and which we denote by $\textbf{Corr}$.  We consider it as the direct sum of spaces of correlation functions of specific Euler characteristic $\chi=-n$, which we denote by $\textbf{Corr}^{\mathbf{(n)}}$:
\begin{equation}
\textbf{Corr}=\overset{\infty}{\underset{n=0}{\oplus}}\textbf{Corr}^{\mathbf{(n)}}
\end{equation}
We have seen already that $\textbf{Corr}$ carries an algebra structure with the $\ast$ product given by (\ref{eq:productcorr}). We see it as a classical space and consider its quantization $\textbf{Corr}_{\mathbf{h}}$ to be the space of correlation functions with $g>0$. In the same way we have
\begin{equation}
\textbf{Corr}_{\mathbf{h}}=\overset{\infty}{\underset{n=0}{\oplus}}\textbf{Corr}^{{\mathbf{(n)}}}_{\mathbf{h}}
\end{equation}
$\textbf{Corr}_{\mathbf{h}}^{\mathbf{(n)}}$ is by definition the free module generated by correlation functions $W^g_k(p,p_1,\dots,p_{k-1})$ with a fixed Euler characteristic $\chi=-n$ and varying genus $g>0$, up to the consistency of the formula $\chi=2-2g-k$, over the ring $k[[]h]]$ of formal power series in a parameter $h$ with coefficients in a field $k$ of characteristic 0. In this way every correlation function $W^{(n)}\in\textbf{Corr}_{\mathbf{h}}^{\mathbf{(n)}}$ admits an expansion in terms of $h$. For $\chi=-n$ even,
\begin{align}
W^{(n)}=&W^0_{n+2}(p,p_1,\dots,p_{n+1})+hW^1_{n}(p,p_1,\dots,p_{n-1})+\notag\\
&\dots +h^{n/2}W^{n/2}_2(p,p_1).
\end{align}
and for $\chi=-n$ odd,
\begin{align}
W^{(n)}=&W^0_{n+2}(p,p_1,\dots,p_{n+1})+hW^1_{n}(p,p_1,\dots,p_{n-1})+\notag\\
&\dots +h^{(n+1)/2}W^{(n+1)/2}_1(p).
\end{align}
We would like to describe a procedure of passing from $\textbf{Corr}$ to $\textbf{Corr}_{\mathbf{h}}$ that encodes Topological Recursion for arbitrary genus. Hopefully this will give some insight to what would be the good direction towards the quantization of the Loday-Ronco Hopf algebra.

First we obtain the one-loop graph of order one, $t^1\in (Y^1)^1$ as a result of applying a quantization map $Q$ to the tree $(\mathbf{1})\in Y^1$:
\begin{equation}
Q(\tree)=\oneloop.
\end{equation}
We also set $Q(|)=0$.
As a concrete model we can take the Teichmuller space of Riemann surfaces, take $(\mathbf{1)}=\tree$ as a pair of pants and $\oneloop$ as the result of attaching a cylinder to two of the borders of $\tree$.
Now we obtain by induction the first order quantization of a tree $t\in Y^n, t=t_1\vee t_2$:
\begin{equation}
Q(t)=Q(t_1)\vee t_2+t_1\vee Q(t_2) +t_1\bridge t_2.
\end{equation}
By successive applications of $Q$ we get graphs with more and more loops starting from a planar binary tree, that correspond to higher order terms in quantization, up to exhausting the available pairs of leaves. Hence the operator $Q$ is nilpotent.
For instance, the quantization of $\onetwo=\tree\vee |$ is
\begin{equation}
Q(\onetwo)=\onetwoloopone+\onetwolooptwo.
\end{equation}
It is clear that $Q^2\left(\onetwo\right)=0$.
We define three more operators $Q_L, Q_M$ and $Q_R$:
\begin{align}
&Q_L(t)=Q(t_1)\vee t_2,\\
&Q_M(t)=t_1\bridge t_2,\\
&Q_R(t)=t_1\vee Q(t_2).
\end{align}
We see that
\begin{equation}
Q(t)=Q_L(t)+Q_M(t)+Q_R(t).
\end{equation}
We would like in the future to study more deeply the quantization operator $Q$ and its derivate operators $Q_L, Q_M$ and $Q_R$.

Next we define a quantum product of two $\tree$ as a first order expansion in $h$:
\begin{align}
\tree\ast_h \tree&= \tree\ast\tree +hQ\left(\tree\ast\tree\right)\notag\\
&=\onetwo +\twoone +h\left(\onetwoloopone+\onetwolooptwo+\twooneloopone+\twoonelooptwo\right).
\end{align}
Then we see that
\begin{equation}
\tree\ast_h \tree= \tree\ast\tree +h\left(\sum_{i=1}^{2} \prescript{}{i}\leftrightarrow_{i+1}\left((\mathbf{1})\ast(\mathbf{1})\right)\right).
\end{equation}
The next step would be to define the product between $\tree$ and $\oneloop$. We do it in a way such that
\begin{equation}
\tree\ast_h \tree=\tree\ast\tree + h (\tree\ast_h\oneloop+\oneloop\ast_h\tree)
\end{equation}
We use the formula (\ref{eq:shuffleident1}) to set
\begin{align}\label{eq:treeoneloop}
\tree\ast_h\oneloop&=|\vee\oneloop+\tree\bridge|\notag\\
&=\twoonelooptwo+\onetwolooptwo.
\end{align}
The formula for $\oneloop\ast_h\tree$ should now be obvious.

Having settled the low order cases we now move to the general case. We take the product of a graph $t^{g_1}\in (Y^p)^{g_1}$ with a graph $t^{g_2}\in (Y^q)^{g_2}$ in such a way as to give the $h^{g_1+g_2}$ contribution to the quantum product $t_1\ast_h t_2$, for $t_1\in Y^p$ and $t_2\in Y^q$ the underlying trees of $t_1^{g_1}$ and $t_2^{g_2}$:
\begin{equation}
t_1\ast_h t_2=t_1\ast t_2+\dots + h^{g_1+g_2}(t_1^{g_1}\ast_h t_2^{g_2}+t_1^{g_2}\ast_h t_2^{g_1})+\dots
\end{equation}
Reversing the argument we see that the quantum product becomes a sum over the genus of graphs:
\begin{equation}
t_1\ast_h t_2=\sum_{k=0}^{(p+q-k_{\text{min}})/2+1}\sum_{l=0}^kh^kt_1^l\ast_h t_2^{k-l}, t_1^l\in (Y^p)^l,t_2^{k-l}\in (Y^q)^{k-l}
\end{equation}
where $k_\text{min}$ is $1\text{ or } 2$.
By summing over all planar binary trees and the corresponding graphs with loops we get the corresponding formula for the quantum product of correlation functions:
\begin{align}
&W^0_{m+2}\ast_h W^0_{n+2}=\notag\\
&\sum_{k=0}^{(m+n-k_{\text{min}})/2+1}\sum_{l=0}^k h^kW^l_{m+2-2l}\ast_h W^{k-l}_{n+2-2(k-l)}
\end{align}
In terms of correlation functions $W^{(n)}\in\textbf{Corr}^{\mathbf{(n)}}_{\mathbf{h}}$ the above examples give
\begin{align}
W^{(1)}&=W^0_3(p,p_1,p_2)+hW^1_1(p)\notag\\
&=\tree + h \oneloop\\
W^{(2)}&=W_4^0(p,p_1,p_2,p_3)+hW^1_2(p,p_1)\notag\\
&=\onetwo+\twoone+h\left(\onetwoloopone+\onetwolooptwo+\twooneloopone+\twoonelooptwo\right)
\end{align}
Now giving the module $\textbf{Corr}^{\mathbf{(n)}}_{\mathbf{h}}$ the obvious ring structure that consists in multiplying term by term the expansion in $h$ we have
\begin{align}
W^{(2)}&=W^{(1)}\cdot W^{(1)}\notag\\
&=W^0_3(p,p_1,p_2)\ast W^0_3(p',p^{'}_1,p^{'}_2)\notag\\
&+h\left(W^1_1(p)\ast_h W^0_3(p',p^{'}_1,p^{'}_2)+W^0_3(p,p_1,p_2)\ast_h W^1_1(p')\right).
\end{align}
This allows to identify the terms in the expansion of $W^{(2)}$:
\begin{align}
&W^0_4(p,p_1,p_2,p_3)=W^0_3(p,p_1,p_2)\ast W^0_3(p,p_2,p_3)\notag\\
&=K_p(q,\bar q)\left(W^0_3(q,p_1,p_2)W^0_2(\bar q,p_3)+W^0_2(q,p_1)W^0_3(\bar q,p_2,p_3)\right);\\
&W^1_2(p,p_1)=W^1_1(p)\ast_h W^0_3(p,p_1,p_2)+W^0_3(p,p_1,p_2)\ast_h W^1_1(p)\notag\\
&=K_p(q,\bar q)\left(W^0_3(q,\bar q,p_1)+W^1_1(q)W^0_2(\bar q,p_1)+W^0_2(q,p_1)W^1_1(\bar q)\right).
\end{align}
where we have used (\ref{eq:treeoneloop}) and the corresponding formula for $\oneloop\ast_h\tree$. Hence the term-by-term product of $W^{(1)}$ by itself gives the corresponding terms in the Topological Recursion formula for the components of $W^{(2)}$.
\bibliographystyle{amsplain}
\bibliography{biblioHopf}
\end{document}